\documentclass[letterpaper, 10 pt, conference]{ieeeconf}
\usepackage{graphicx}
\usepackage{mathtools}
\usepackage{array}
\usepackage{amsfonts}
\usepackage{caption}

\usepackage{amssymb, amsthm}

\theoremstyle{plain}
\newtheorem{thm}{Theorem}
\newtheorem{cor}{Corollary}
\newtheorem{prop}{Proposition}

\newtheorem{lemma}{Lemma}

\theoremstyle{definition}
\newtheorem{defn}{Definition}

\newtheorem{ex}{Example}

\theoremstyle{remark}

\renewcommand{\rm}[1]{\mathrm{#1}}
\newcommand{\bb}[1]{\mathbb{#1}}
\newcommand{\cl}[1]{\mathcal{#1}}
\newcommand{\R}{\bb{R}}

\newcommand{\card}[1]{| #1 |}

\newcommand{\ac}{\mathcal{A}} 
\newcommand{\p}{\mathcal{I}} 
\newcommand{\w}{\mathrm{W}} 
\newcommand{\ui}{\mathrm{U}_i} 
\newcommand{\G}{\mathrm{G}} 
\newcommand{\aopt}{a^{\rm{opt}}} 
\newcommand{\ane}{a^{\rm{ne}}} 
\newcommand{\br}{\rm{Br}_i} 

\newcommand{\abri}{a_\rm{br}^i}  
\newcommand{\poa}{\rm{PoA}}  
\newcommand{\pose}{\rm{PoSE}}  
\IEEEoverridecommandlockouts             
\title{Quality of Non-Convergent Best Response Processes in Multi-Agent Systems through Sink Equilibrium
}

\author{Rohit Konda \and Rahul Chandan \and Jason R. Marden%
\thanks{R. Konda (\texttt{rkonda@ucsb.edu}), R. Chandan, and J. R. Marden are with the Department of Electrical and Computer Engineering at the University of California, Santa Barbara, CA. This work is supported by \texttt{ONR grant \#N00014-20-1-2359} and \texttt{AFOSR grants \#FA9550-20-1-0054} and \texttt{\#FA9550-21-1-0203}.}%
}

\begin{document}

\maketitle
\thispagestyle{empty}
\pagestyle{empty}

\begin{abstract}
Examining the behavior of multi-agent systems is vitally important to many emerging distributed applications - game theory has emerged as a powerful tool set in which to do so. The main approach of game-theoretic techniques is to model agents as players in a game, and predict the emergent behavior through the relevant Nash equilibrium. The virtue from this viewpoint is that by assuming that self-interested decision-making processes lead to Nash equilibrium, system behavior can then be captured by Nash equilibrium without studying the decision-making processes explicitly. This approach has seen success in a wide variety of domains, such as sensor coverage, traffic networks, auctions, and network coordination. However, in many other problem settings, Nash equilibrium are not necessarily guaranteed to exist or emerge from self-interested processes. Thus the main focus of the paper is on the study of sink equilibrium, which are defined as the attractors of these decision-making processes. By classifying system outcomes through a global objective function, we can analyze the resulting approximation guarantees that sink equilibrium  have for a given game. Our main result is an approximation guarantee on the sink equilibrium through defining an introduced metric of misalignment, which captures how uniform agents are in their self-interested decision making. Overall, sink equilibrium are naturally occurring in many multi-agent contexts, and we display our results on their quality with respect to two practical problem settings.
\end{abstract}

\section{Introduction}
\label{sec:int}
The analysis of multi-agent systems has received a significant amount of attention recently, primarily due to the emergence of distributed structures in wireless communication, biology, IOT, and many other application domains. This has led to a rich variety of theoretical approaches \cite{kubisch2003distributed, babaoglu2006design, stolpe2016internet}. In this paper, we specifically consider a \emph{game theoretic} approach, where the emergent properties of the multi-agent system are studied using tools from game theory. 

The main idea of this approach is to model agents as self-interested decision makers, where each agent's preference over the collective system outcome is designated through a \emph{utility function}. The agents are then presumed to undergo a \emph{best (or better) response process}, where each agent updates their decision in a self-interested manner to maximize their individual utility. The emergent system outcomes from this process are traditionally expected to be \emph{Nash equilibrium}, which can be expressed as the limit points of the best response process. Thus, many previous works study the emergent system behavior through characterizing properties of the Nash equilibrium; this has been done in many different contexts, such as traffic systems, power networks, etc. \cite{roughgarden2005selfish, atzeni2013noncooperative}.

\textbf{But is it reasonable to expect agents to converge to Nash equilibrium from best response processes?} In the class of \emph{potential games} \cite{monderer1996potential} and variants \cite{marden2007regret}, best response processes are indeed guaranteed to converge to Nash equilibrium \cite{swenson2018best}. In potential games, agents are fully cooperative, and self-interested decisions made by agents always lead to improvements in a given global objective. However, a variety of natural multi-agent settings fall outside of this class. The system may display competitive interactions; for e.g., business firms may have competing economic interests. In social systems, agents' may be inherently misaligned in their preferences; for e.g., drivers may have different sensitivities to tolls. Even for multi-agent systems that are fully engineered, there are operational concerns such as prediction errors or informational privacy that must be accounted for. All of these scenarios do not exhibit a potential game structure, and thus guarantees of convergence to Nash equilibrium can not be established. We describe two examples of this kind in Example \ref{ex:env} and \ref{ex:radio}. \textbf{In these instances, can the emergent system outcomes still be characterized, where either Nash equilibrium do not exist or best response processes does not converge to Nash equilibrium?} This is the main focus of this paper.

We utilize \emph{sink equilibrium} in \cite{goemans2005sink} as an alternate solution concept \footnote{A popular alternative to Nash equilibrium are \emph{coarse correlated equilibrium} which have existence and convergence guarantees \cite{hart2000simple}. However, these dynamics requires full knowledge of the history of decisions of all of the agents. Additionally, only empirical distribution of play is guaranteed to the converge to coarse correlated equilibrium, which may not correspond to the actualized decisions made by the agents.} to address this, which are specified as the attractors of the best response process. By definition, sink equilibrium are well defined and have guarantees of convergence for any given game. Thus, we analyze the behavior of sink equilibrium in this paper. Specifically, we assume that system outcomes are evaluated through a given global objective function. The utility function of each agent may not be aligned with this global objective. In these settings, we can characterize performance guarantees of the induced sink equilibrium with respect to the global objective - this is the main result of the paper.

To the author's knowledge, a general approach to studying performance guarantees of sink equilibrium has not been done previously. While sink equilibrium have not been studied in as much detail as Nash equilibrium, we still highlight an pertinent selection of past literature on sink equilibrium. The seminal work in \cite{goemans2005sink} first established the concept of sink equilibrium in a game-theoretic context, as well as provide negative results of sink equilibrium in valid utility games. Positive results on sink equilibrium behavior were shown in \cite{kleinberg2011beyond}, where it was seen that sink equilibrium perform much better than mixed Nash equilibrium. Analysis of sink equilibrium under the name of curb sets was done in \cite{basu1991strategy}. Several complexity results of sink equilibrium were introduced in \cite{fabrikant2008complexity}. Sink equilibrium were extended in continuous domains in \cite{papadimitriou2019game}. Recently, design of sink equilibrium selection algorithms was done in \cite{yan2021policy}. While the literature on sink equilibrium is sparse, they naturally emerge when agent utility functions are not aligned perfectly. 

The structure of the paper is as follows. In Section \ref{sec:model}, we define our game theoretic setup and formally define the best response process and sink equilibrium. Additionally, we introduce two application scenarios where sink equilibrium appear. In Section \ref{sec:results}, we state our main results, where we derive performance guarantees of sink equilibrium in our setting with regards the approximation ratio of \emph{price of sinking}. A discussion of sink equilibria induced by better response processes is provided in Section \ref{sec:better}. Finally, we conclude in Section \ref{sec:conc}. We include certain proofs in the Appendix.

\section{Preliminaries}
\label{sec:model}

Consider a general multi-agent scenario with $n$ agents $\p=\{1, \dots, n\}$, where each agent is endowed with a decision or action set $\ac_i$. We denote an action as $a_i \in \ac_i$, and a joint action profile as $a \in \ac = \ac_1 \times \cdots \times \ac_n$. The quality of each joint action profile is evaluated with a global objective function $\w: \ac \to \bb{R}_{\geq 0}$ that characterizes the total system welfare. In other words, the optimal joint decision of the system is described as
\begin{equation} \label{eq:opt}
    \aopt \in \arg \max_{a \in \ac} \w(a).
\end{equation}
We assume that solving the decision problem in Eq. \eqref{eq:opt} cannot be done in a centralized fashion, due to computational, informational, administrative, or communication concerns. Therefore, we assume that each agent selects their decisions in a distributed manner. We assume that each agent $i$ is endowed with a \emph{utility function} $\ui: \ac \to \R$ to classify their preferences over their decision set, resulting in the game tuple $\G \triangleq (\p, \ac, \w, \{\ui\}_{i \in \p})$. Moreover, we assume the utility function $\ui$ depends on the welfare function $\w$, either naturally or by design. In particular, consider a completely cooperative scenario, where the agents exhibit the \emph{common interest} utility design. Here, the utility functions are completely aligned with the global objective $\w$ \footnote{We can relax this constraint to consider utility functions that are \emph{preference equivalent} to the welfare function, where the ordering of preferences over joint actions is maintained. This is indeed the case for potential and weighted-potential games.}, where 
\begin{equation}
    \ui(a) \equiv \w(a) \text{ for all } a \in \ac \text{ and } i \in \p.
\end{equation}
Under the common interest utility, the emergent joint decisions coincide with the set of Nash equilibrium $\rm{NE}$ of the game. A joint action $\ane$ is considered a Nash equilibrium if the following inequality holds
\begin{equation}
    \ui(\ane) \geq \ui(a_i, \ane_{-i}) \text{ for all } a \in \ac \text{ and } i \in \p,
\end{equation}
where $a_{-i} = (a_1, \dots, a_{i-1}, a_{i+1}, \dots, a_n)$ corresponds to the joint action without the decision of agent $i$. In completely cooperative settings, where they are guaranteed to exist, we can then quantify the emergent behavior through the qualities of the possible resulting Nash equilibrium. This is done through the metric of \emph{price of anarchy} which is defined as
\begin{equation}
\label{eq:poa}
\poa(\G) = \frac{\min_{\ane \in \rm{NE}} \w(\ane)}{\max_{a \in \ac} \w(a)},
\end{equation}
where we take the worst case ratio of the welfare of Nash equilibrium over the optimal welfare. Price of anarchy is a well-studied metric, with many results on its characterization in the literature \cite{roughgarden2015intrinsic, paccagnan2021utility, roughgarden2005selfish}. Thus a standardized approach can be implemented to characterize system behavior in completely cooperative scenarios.

However, due to operational or design constraints, assuming a common interest utility design may not be feasible (see Examples \ref{ex:env} and \ref{ex:radio}). Throughout the paper, we then allow $\ui(a) \neq \w(a)$ to be misaligned, and focus our attention to sink equilibrium as our standard solution concept. To define sink equilibrium, we first outline the \emph{best response process}. Under this process, the best decision set for agent $i$ assuming all other agents' decisions are fixed is known as the \emph{best response set}, that is,
\begin{equation}
\label{eq:brset}
    \br(a) = \arg \max_{\bar{a}_i \in \ac_i} \ui(\bar{a}_i, a_{-i}).
\end{equation}
Then for every step, a randomly selected agent picks an action from its best response set uniformly. This induces the following Markovian dynamics on the set of joint actions, describing the best response process, as 
\begin{equation}
\label{eq:probrandbest}
    \rm{Pr}(\tilde{a} | a) = \begin{cases} 
    \frac{1}{n \cdot \card{\br(a)}} &\text{if } \tilde{a} \in \br(a) \text{ for some } i  \\
    0 & \text{otherwise},
    \end{cases}
\end{equation}
where $\rm{Pr}(\tilde{a} | a)$ represents the probability of reaching the joint action $\tilde{a}$ from $a$ in the Markov chain. Note that there is an equal chance for each player $i$ to perform a best response at each time step. We refer to a probability distribution over the joint action set as $\sigma \in \Delta \ac$, where $p_{\sigma}(a)$ denotes the probability of sampling $a$ under $\sigma$. We also say that an action $a \in \rm{supp}(\sigma) \subseteq \ac$ is in the support of $\sigma$ if the probability $p_{\sigma}(a) > 0$ is strictly positive. Furthermore, we say that $\sigma$ is a \emph{stationary distribution} of the Markov chain if the equality $p_{\sigma}(a) = \bb{E}_{\bar{a} \sim \sigma}[\rm{Pr}(\tilde{a} | \bar{a}) \cdot p_{\sigma}(\bar{a})]_a$ for all $a \in \ac$ holds. We also specify the \emph{sink strongly connected components} of the Markov chain defined in Eq. \eqref{eq:probrandbest}. A set $S \subseteq \ac$ is a sink strongly connected component if there is exists a path of positive probability from $a$ to $\bar{a}$ under \eqref{eq:probrandbest} for any $a, \bar{a} \in S$ and there are no transitions from $S$ to outside of $S$. Formally, if $a, \bar{a} \in S$, then there exists a sequence $a_0, a_1, \dots, a_m$ where $\rm{Pr}(a_{j+1} | a_j) > 0$ for all $j$ with $a_0 = a$ and $a_m = \bar{a}$. Additionally, if $a \in S$ and $\bar{a} \notin S$, then no such sequence exists. 

\begin{defn}
A probability distribution $\sigma \in \Delta \ac$ is a sink equilibrium of the game $\G$ if $\sigma$ is a stationary distribution of the Markov chain in Eq. \eqref{eq:probrandbest} and if $\rm{supp}(\sigma) = S$ is a sink strongly connected component.
\end{defn}

In other words, sink equilibria are defined as the attractors of the best response process. Given a game $\G$, we characterize the behavior of the sink equilibria in a similar manner to Eq. \eqref{eq:poa} through the metric of \emph{price of sinking} \cite{goemans2005sink} as
\begin{equation}
    \pose(\G) = \min_{\sigma \in \rm{SE}} \bb{E}_{a \sim \sigma} [\w(a)] / \w(\aopt),
\end{equation}
where $\rm{SE}$ denotes the total set of sink equilibrium of the game $\G$. We note that under a common interest utility design, the set of Nash equilibrium $\rm{NE} \simeq \rm{SE}$ is equivalent to the set of sink equilibrium.

We now examine the sink equilibrium on two illustrative examples in which they naturally emerge. In these examples, we derive guarantees on the sink equilibrium using the tools discussed in Section \ref{sec:results}. The examples are as follows.

\begin{ex}[Ecological Monitoring]
\label{ex:env}
Ecological monitoring is necessary to understand the well-being of inhabited populations as well as track health of overall ecology. While this can be handled by field ecologists, autonomous agents can supplement or even act as substitutes to gather important ecological data - as was done by the authors in \cite{notomista2019slothbot}. In this scenario, a high level control objective of the agents is to understand how to orient themselves to monitor the region of interest as best as possible. We can model this as a covering problem \cite{chvatal1979greedy}. In this way, let $\cl{R} = \{r_1, \dots, r_m\}$ define the region monitored by $n$ agents, where we have finitely partitioned the region into possibly arbitrary segments. For each segment $r$, the importance of monitoring that segment can be associated with a value $v_r \in \R_{\geq 0}$ which defines the intrinsic quality of data that can be collected in that segment. This can be affected by the number or magnitude of populations in the segment, relevant climactic conditions, etc. These parameters are never known before hand, and thus must be estimated by the agents in the field. Thus each agent has its noisy estimate of the value $v^i_r$ which we assume is drawn from a normal distribution $\cl{N}[v_r + c, (d \cdot v_r)^2]$ with bias $c$ and variance $(d \cdot v_r)^2$. Each agent can decide which subset of region to monitor (i.e. $a_i \subset \cl{R}$), which depend on its sensor and motor capabilities. Thus, collectively, the goal of the agents is to monitor the most and highest valued portions, as captured by the welfare function below.
\begin{equation}
    \label{eq:setcov}
    \w(a) = \sum_{r \in \cup_i a_i} v_r
\end{equation}
The objective that each agent witnesses, however is based on their estimate, or that their utility is $\ui(a) =\sum_{r \in \cup_i a_i} v^i_r$, which is potentially different for each agent. Thus, when running a best response algorithm, the agents converge to a sink equilibrium which may not be Nash equilibrium. We derive a guarantee on price of sinking in  Proposition \ref{prop:ex1}.

\begin{prop}
\label{prop:ex1}
Consider the problem defined above. The expected price of sinking is lower bounded by the following expression
\begin{equation}
    \bb{E}[\pose(\G)] \geq \max(\frac{1 - 4 n \beta_\Phi }{2}, 0),
\end{equation}
where $\Phi$ is the normal cumulative distribution function and $\beta_\Phi$ is defined as
\begin{equation}
\label{eq:misex1}
    \beta_\Phi = |\cl{R}| \left( d \sqrt{\frac{2}{\pi}}e^{-\frac{c^2}{2d^2}} + c(1 - 2 \Phi(-c/d)) \right).
\end{equation}
\end{prop}
\begin{proof}
Proof can be found in the appendix.
\end{proof}
\end{ex}

\begin{ex}[Radio Signalling]
\label{ex:radio}
Consider the situation in which $n$ agents have to communicate to each other through $k$ communication channels, as seen in \cite{gourves2009strong}. However, when more than one agent selects a channel to communicate under, the signals experience interference. This can captured by the parameter $w_{ij} \geq 0$, which dictates the interference between agent $i$ and agent $j$. The system as a whole, would like to minimize the total interference experienced between all the agents. In turn, the system welfare can be defined as
\begin{equation}
\label{eq:welfrad}
    \w(a) = \sum_i \sum_{j: a_j \neq a_i} w_{ij},
\end{equation}
where $a_i \in \{1, \dots k\}$ is the channel that agent $i$ decides to transmit their messages on. However, these interference parameters may not be known to the agents and have to be estimated. For simplicity, we can assume that the agents have a margin of error of $\alpha$ or that agent i's estimate is $w_{ij}^{i} \in [\alpha \cdot w_{ij}, \alpha^{-1} \cdot w_{ij}]$. Again, when agents run a best response algorithm, they are not guaranteed to converge to a Nash equilibrium due to the informational constraints. Then we can characterize the guarantee on the price of sinking in Proposition \ref{prop:ex2}.

\begin{prop}
\label{prop:ex2}
Consider the problem defined above with two channels ($k=1$). The price of sinking is lower bounded by the following expression
\begin{equation}
    \label{eq:misex2}
    \pose(\G) \geq \frac{1}{3 \alpha^2 + (1 - \alpha^2)n}.
\end{equation}
\end{prop}
\begin{proof}
Proof can be found in the appendix.
\end{proof}
\end{ex}

\section{Main Results}
\label{sec:results}

The main results of this paper are on providing lower bounds for the price of sinking for a given game. To do this, we first recall the notion of \emph{smoothness} \cite{roughgarden2015intrinsic} as a useful analysis tool for the price of anarchy. In this paper, we use a relaxed version of smoothness to classify a given game. We say that a game is \emph{$(\lambda, \mu)$-smooth} if, for a fixed $\mu \geq \lambda \geq 0$, we have that 
\begin{equation}
\label{eq:smooth}
    \sum_{i \in \p} \left( \ui(a)  - \ui(\aopt_i, a_{-i}) \right) \leq \mu \w(a) - \lambda \w(\aopt)
\end{equation}
for all actions $a \in \ac$. Given these parameters, the efficiency of Nash equilibrium for a given game can easily be determined. This is described in Proposition \ref{prop:lambdamu}, where a proof is included for completeness.
\begin{prop}
\label{prop:lambdamu}
Let $\G$ be a $(\lambda, \mu)$-smooth game. Then the price of anarchy is lower bounded by $\poa(\G) \geq \frac{\lambda}{\mu}$.
\end{prop}
\begin{proof}
From applying the definition of Nash equilibrium repeatedly for all $i \in \p$ with respect to the deviation $\aopt$, we have the following inequality
\begin{equation*}
    \sum_i \left( \ui(\ane) - \ui(\aopt_i, \ane_{-i}) \right) \geq 0.
\end{equation*}
Notice that now we can directly substitute the inequality in Eq. \eqref{eq:smooth} to get that
\begin{equation*}
    \mu \w(\ane) - \lambda \w(\aopt) \geq 0.
\end{equation*}
Since $\ane$ is any arbitrary Nash equilibrium, we can rearrange the above equation to get that $\poa(\G) \geq \frac{\lambda}{\mu}$ to show the claim.
\end{proof}

We note that there always exist some $\lambda$ and $\mu$ such that the game is $(\lambda, \mu)$-smooth, as $\lambda \to 0$ and $\mu \to \infty$ will always satisfy the inequality in Eq. \eqref{eq:smooth}. The main analytical benefit of smoothness analysis is that instead of searching across the set of Nash equilibrium directly, we can instead characterize the price of anarchy through a bi-variable optimization problem (over $\lambda$ and $\mu$). This can be done in various game-theoretic contexts \cite{roughgarden2015intrinsic}. The optimization problem is written formally below.

\vspace{1em}
\begin{cor}
The price of anarchy for a given game $\G$ is lower bounded by
\begin{equation*}
    \rm{PoA}(\G) \geq \sup_{\mu \geq \lambda \geq 0} \{\frac{\lambda}{\mu}: \G \text{ is } (\lambda, \mu)-\text{smooth}\}.
\end{equation*}
\end{cor}

However, unlike Nash equilibrium, when applying the smoothness inequality directly to the analysis of sink equilibrium, it is not possible to get non-trivial guarantees. For all valid smoothness parameters, it is possible to construct a corresponding game with trivial performance guarantees on sink equilibrium, as stated below.

\begin{prop}
\label{prop:bestworst}
For every $\mu > \lambda \geq 0$, there exists a $(\lambda, \mu)$-smooth game $\G$ with a unique sink equilibrium such that the price of sinking is $\pose(\G) = 0$.
\end{prop}

\begin{proof}
Let $\mu > \lambda \geq 0$. Consider the following game $\G$ with two agents with the action sets $\ac_1 = \{e_1, e_2, e_3\}$ and $\ac_2 = \{f_1, f_2, f_3\}$. We define the welfare values $\w(a)$ for each joint action through Table \ref{tab:bestwelf} below.
\begin{table}[h!]
    \centering
    \begin{tabular}{llll}
        & \multicolumn{1}{c}{$f_1$} & \multicolumn{1}{c}{$f_2$} & \multicolumn{1}{c}{$f_3$} \\ \cline{2-4} 
    \multicolumn{1}{l|}{$e_1$} & \multicolumn{1}{l|}{1}                             & \multicolumn{1}{l|}{$(\lambda + \varepsilon)/\mu$} & \multicolumn{1}{l|}{0}    \\ \cline{2-4} 
    \multicolumn{1}{l|}{$e_2$} & \multicolumn{1}{l|}{$(\lambda + \varepsilon)/\mu$} & \multicolumn{1}{l|}{0}                             & \multicolumn{1}{l|}{0}    \\ \cline{2-4} 
    \multicolumn{1}{l|}{$e_3$} & \multicolumn{1}{l|}{0}                             & \multicolumn{1}{l|}{0}                             & \multicolumn{1}{l|}{0}    \\ \cline{2-4} 
    \end{tabular}
\caption{Welfare $\w(a)$ for each joint action $a = (e_i, f_j)$.}
\label{tab:bestwelf}
\end{table}
Similarly, we can define the utility values $\ui(a)$ for each joint action and for each agent in Table \ref{tab:bestutil}.
\begin{table}[h!]
    \centering
    \begin{tabular}{llll}
     & \multicolumn{1}{c}{$f_1$}                & \multicolumn{1}{c}{$f_2$}                     & \multicolumn{1}{c}{$f_3$}                    \\ \cline{2-4} 
    \multicolumn{1}{l|}{$e_1$} & \multicolumn{1}{l|}{$(0, 0)$}            & \multicolumn{1}{l|}{$(0, \varepsilon)$}       & \multicolumn{1}{l|}{$(0, -\varepsilon)$}     \\ \cline{2-4} 
    \multicolumn{1}{l|}{$e_2$} & \multicolumn{1}{l|}{$(\varepsilon, 0)$}  & \multicolumn{1}{l|}{$(\lambda, -2 \lambda)$}  & \multicolumn{1}{l|}{$(-2 \lambda, \lambda)$} \\ \cline{2-4} 
    \multicolumn{1}{l|}{$e_3$} & \multicolumn{1}{l|}{$(-\varepsilon, 0)$} & \multicolumn{1}{l|}{$(- 2 \lambda, \lambda)$} & \multicolumn{1}{l|}{$(\lambda, -2 \lambda)$} \\ \cline{2-4} 
    \end{tabular}
    \caption{Welfare $(\rm{U}_1(a), \rm{U}_2(a))$ for each joint action $a$.}
    \label{tab:bestutil}
\end{table}

We can choose $\varepsilon = (\mu - \lambda)/2 > 0$ such that the optimal joint action is $\aopt = (e_1, f_1)$ with an optimal welfare of  $\w(\aopt) = 1$. Under the best response dynamics, observe that the set $\{(e_2, f_2), (e_3, f_2), (e_2, f_3), (e_3, f_3)\}$ is the unique strongly connected component. It can be verified that each joint action $a$ satisfies the smoothness condition in Eq. \eqref{eq:smooth}. Since the welfare of each action in the unique strongly connected component is $0$, the price of sinking can be upper bounded by $\pose(\G) = \bb{E}_{a \in \sigma} [\w(a)] \leq \sum_{a \in \rm{supp}(\sigma)} \w(a) = 0$ for the unique sink equilibrium $\sigma$.
\end{proof}

This negative result is similar in spirit to the one presented in \cite[Lemma 3.2]{goemans2005sink}. However, we emphasize that the inferior guarantees are more indicative of inefficacy of a direct approach rather than the intrinsic behavior of sink equilibrium. This sentiment is also reflected in \cite{kleinberg2011beyond}, where in certain game settings, it is shown that the quality of sink equilibrium is arbitrarily better than the quality of any mixed equilibrium. In fact, if we consider games with added structure, we can arrive at nontrivial guarantees on sink equilibrium.

Therefore, we consider games in which the deviation from the common interest utility $\ui(a) \neq \w(a)$ is bounded. We encapsulate the extent of the deviation through the constant $\beta \in [0, 1]$, where $\beta = 0$ signifies no deviation from the common interest utility and $\beta = 1$ signifies the maximum deviation. We define this formally below  \footnote{We can generalize the results to instead consider alignment to a potential function. In this way, $\beta$ characterizes the closeness of the game to a potential game. Near-potential games have been studied in \cite{candogan2011flows}.}.

\begin{defn}
A game $\G$ is considered to be \emph{$\beta$-arithmetically misaligned} if
\begin{equation}
\label{eq:arthmis}
    | \ui(a) -  \w(a) | \leq \beta \w(a),
\end{equation}
or \emph{$\beta$-geometrically misaligned} if
\begin{equation}
\label{eq:geomis}
    1 -\beta \leq \frac{\ui(a)}{\w(a)} \leq \frac{1}{1 -\beta},
\end{equation}
is satisfied for all actions $a \in \ac$ and agents $i \in \p$.
\end{defn}
 
We note that when $\beta=0$, under the common interest utility, the sink equilibrium are equivalent to the Nash equilibrium and inherit the price of anarchy guarantees coming from smoothness analysis. Likewise, we will see that the sink equilibrium in near-common interest games with $\beta$ close to $0$ inherit similar guarantees dictated by the common interest utility. This observation is also reflected in a different context in \cite{candogan2013dynamics}. In this vein, let $\lambda_c$ and $\mu_c$ be the parameters that satisfy the smoothness inequality in Eq. \eqref{eq:smooth} for the common interest utility
\begin{equation}
\label{eq:comsmooth}
    \sum_{i \in \p} \left( \w(a)  - \w(\aopt_i, a_{-i}) \right) \leq \mu_c \w(a) - \lambda_c \w(\aopt),
\end{equation}
where we have substituted $\ui(a) \equiv \w(a)$. With this, we can state the main result of the paper.

\begin{thm}
\label{thm:main}
Let $\G$ be a game such that the best response $\br(a)$ is always singular valued. Let $\lambda_c$ and $\mu_c$ satisfy Eq. \eqref{eq:comsmooth} for all $a \in \ac$. If the game is $\beta$-arithmetically misaligned, as in Eq. \eqref{eq:arthmis}, then the price of sinking satisfies
\begin{equation}
\label{eq:poseart}
    \pose(\G) \geq \max (\frac{\lambda_c - 4 \beta n}{\mu_c}, 0).
\end{equation}
If the game is $\beta$-geometrically misaligned, as in Eq. \eqref{eq:geomis}, then the price of sinking satisfies
\begin{equation}
\label{eq:posegeo}
        \pose(\G) \geq \frac{\lambda_c}{(1-\beta)^2 \mu_c + (1 - (1 - \beta)^2)n}.
\end{equation}
\end{thm}
\begin{proof}
First, we introduce the following lemma to characterize the sink equilibrium in an alternative fashion.
\begin{lemma}
\label{lem:gsink}
Let $\G$ be a game such that $\br(a)$ is always singular valued and let $\sigma \in \rm{SE}$ be any sink equilibrium of the game. For any function $g: \ac \to \bb{R}$, the following equality must hold
\begin{equation}
    \bb{E}_{a \sim \sigma} [\sum_{i \in \p} g(a) - g(\br(a), a_{-i})] = 0.
\end{equation}
\end{lemma}
\begin{proof}
Let $\sigma$ be a sink equilibrium of the game $\G$. Since the sink equilibrium is a stationary distribution under the dynamics outlined in Eq. \eqref{eq:probrandbest}, we have that $p_{\sigma}(a) = \sum_{\bar{a} \in \ac} \rm{Pr}(a | \bar{a}) p_{\sigma}(\bar{a})$. Under this statement, we have the series of equalities below
\begin{align*}
    n \sum_a p_{\sigma}(a) g(a) &= n \sum_a \sum_{\bar{a}} \rm{Pr}(a | \bar{a}) p_{\sigma}(\bar{a}) g(a) \\
    \bb{E}_{a \sim \sigma} [\sum_{i \in \p} g(a)] &= \bb{E}_{\bar{a} \sim \sigma}[n \sum_{a} \rm{Pr}(a | \bar{a}) g(a)] \\
    &= \bb{E}_{a \sim \sigma}[\sum_{i \in \p} g(\br(a), a_{-i})],
\end{align*}
where we change the the naming convention from $\bar{a}$ to $a$ in the last line. Rearranging the terms and using linearity of expectation gives us the claim.
\end{proof}
\begin{proof}[Proof of Arithmetic]
For ease of notation, let $\abri = (\br(a), a_{-i})$. We can apply Lemma \ref{lem:gsink} with respect to the welfare function $\w$ to get
\begin{equation}
\label{eq:breqzero}
        \bb{E}_{a \sim \sigma} [\sum_{i \in \p} \w(a) - \w(\abri)] = 0.
\end{equation}

Since we assume the game is $\beta$-arithmetically misaligned, we have that $\ui(a) \geq (1-\beta)\w(a) \geq \w(a) - \beta \w(\aopt)$. Likewise, we can also bound $\ui(\abri) \leq \w(\abri) + \beta \w(\aopt)$. We can substitute these two inequalities in Eq. \eqref{eq:breqzero} to get
\begin{equation*}
    \w(a) - \w(\abri) \leq  \ui(a) - \ui(\abri) +  2\beta \w(\aopt)
\end{equation*}
We can apply this inequality to Eq. \eqref{eq:breqzero} for
\begin{equation*}
        \bb{E}_{a \sim \sigma} [2\beta n \w(\aopt) + \sum_{i \in \p} \ui(a) - \ui(\abri)] \geq 0.
\end{equation*}
Further, observe that since $\ui(\abri) \geq \ui(\aopt_i, a_{-i})$ from the definition of a best response, we can replace $\ui(\abri)$ $\ui(\aopt_i, a_{-i})$. We can utilize the $\beta$-misalignment and substitute for the utility functions similarly as before.
\begin{equation*}
    \bb{E}_{a \sim \sigma} [4 \beta n \w(\aopt) + \sum_{i \in \p} \w(a) - \w(\aopt_i, a_{-i})] \geq 0.
\end{equation*}
Applying the definition of $\lambda_c$ and $\mu_c$ as in Eq. \eqref{eq:comsmooth} results in the final inequality.
\begin{equation*}
    \bb{E}_{a \sim \sigma} [4 \beta n \w(\aopt) + \mu_c \w(a) - \lambda_c \w(\aopt)] \geq 0.
\end{equation*}
Notice that the above inequality holds for any arbitrary sink equilibrium $\sigma$. Thus rearranging terms and using linearity of expectation gives us the price of sinking guarantee in Eq. \eqref{eq:poseart}.
\end{proof}

\begin{proof}[Proof of Geometric]
We can apply Lemma \ref{lem:gsink} with respect to the welfare function $\w$ to get Eq. \eqref{eq:breqzero}. For ease of notation, let $\bar{\beta} = 1 - \beta$. We can successively apply the geometric misalignment property in Eq. \eqref{eq:geomis}, as well as using the fact that $\ui(\abri) \geq \ui(\aopt_i, a_{-i})$, to arrive at the following set of inequalities.
\begin{equation*}
    W(\abri) \geq \bar{\beta} \ui(\abri) \geq \bar{\beta} \ui(\aopt_i, a_{-i}) \geq \bar{\beta}^2 \w(\aopt_i, a_{-i})
\end{equation*}
Substituting these inequalities back into Eq. \eqref{eq:breqzero} gives
\begin{equation*}
    \bb{E}_{a \sim \sigma}[\sum_i \w(a) - \bar{\beta}^2 \w(\aopt_i, a_{-i})] \geq 0.
\end{equation*}
Now we can substitute the definitions of $\lambda_c$ and $\mu_c$ in Eq. \eqref{eq:comsmooth} to a portion of the terms and simplify to get
\begin{equation*}
    \bb{E}_{a \sim \sigma}[n (1 - \bar{\beta}^2)\w(a) + \bar{\beta}^2 \left( \mu_c \w(a) - \lambda_c \w(\aopt) \right)] \geq 0.
\end{equation*}
Notice that the above inequality holds for any arbitrary sink equilibrium $\sigma$. Thus rearranging terms and using linearity of expectation gives us the price of sinking guarantee in Eq. \eqref{eq:posegeo}.
\end{proof}
Thus we have shown the bounds for both geometric and arithmetic misalignment cases.
\end{proof}

We see that when the utility functions are close to the common interest utility design with $\beta \sim 0$, the price of sinking guarantees match the guarantees for the common interest utility. We note that while our approach allows us to get nontrivial guarantees on the sink equilibrium, we still suffer from the degradation of the guarantee as the number of agents $n \to \infty$ increases arbitrarily. However, we assume worst case deviations (see $\w(\abri) \geq \bar{\beta}^2 \w(\aopt_i, a_{-i})$ in the proof of the geometric misalignment) for all actions in the game, which is not true for most natural games and produces a conservative bound. Thus the focus of future work will be to address this concern to get tighter guarantees. We can also get alternative guarantees if we consider sink induced by better responses. This is discussed in the next section.

\section{Sink Eq. From Better Responses}
\label{sec:better}
In this section we consider sink equilibrium that are induced by a better (rather than best) response process. In contrast to the best response set in Eq. \eqref{eq:brset}, we consider the \emph{better response set} defined as
\begin{equation}
    \rm{br}_i(\bar{a}) = \{a_i \in \ac_i: \ui(a_i, \bar{a}_{-i}) \geq \ui(\bar{a}) \}
\end{equation}
for a given agent $i$. The better response process is then defined by a random walk, similar to Eq. \eqref{eq:probrandbest} as 
\begin{equation}
\label{eq:probrandbetter}
    \rm{Pr}(\tilde{a} | a) = \begin{cases} 
    \frac{1}{n \cdot \card{\rm{br}_i(a)}} &\text{if } \tilde{a} \in \rm{br}_i(a) \text{ for some } i  \\
    0 & \text{otherwise}.
    \end{cases}
\end{equation}
The sink equilibrium are similarly defined for these dynamics. If we consider sink equilibrium that are induced by better responses, it is possible to get positive guarantees on the behavior of sink equilibrium. More specifically, we show that there always exists a joint action in the support of the sink equilibrium that has similar guarantees to the Nash equilibrium.

\begin{prop}
\label{prop:better}
Let $\G$ be $(\lambda, \mu)$-smooth. Every sink equilibrium in $\G$ induced by better responses contains a joint action $\tilde{a} \in \rm{supp}(\sigma)$ in its support such that $\w(\tilde{a}) \geq \frac{\lambda}{\mu} \w(\aopt)$.
\end{prop}
\begin{proof}[Proof of Proposition \ref{prop:better}]
We show that $\w(a) \geq \frac{\lambda}{\mu} \w(\aopt)$ for some $a \in \rm{supp}(\sigma)$ in the sink induced by better responses. We first claim that for any sink $\sigma$, there exists a joint action $a \in \rm{supp}(\sigma)$ such that for all $i \in \p$,
\begin{equation}
\label{eq:brsinksmooth}
    \ui(a) - \ui(\aopt_i, a_{-i}) \geq 0.
\end{equation}
Consider an arbitrary action $a \in \rm{supp}(\sigma)$ in which the condition does not hold true. Let $i$ be the smallest number such that $a$ does not satisfy Eq. \eqref{eq:brsinksmooth} for agent $i$. Then the action $\hat{a} = (\aopt_i, a_{-i})$ is a better response to $a$ and therefore $\hat{a} \in \rm{supp}(\sigma)$ is in the support of $\sigma$ as well. Note that $\hat{a}$ also satisfies Eq. \eqref{eq:brsinksmooth} for agent $i$. By induction, we can then derive an action $a^* \in \rm{supp}(\sigma)$ such that Eq. \eqref{eq:brsinksmooth} is satisfied for all $i$. By the smoothness inequality in Eq. \eqref{eq:smooth}, we have that
\begin{equation*}
    \mu \w(a^*) - \lambda \w(\aopt) \geq \sum_i \ui(a^*)  - \ui(\aopt_i, a^*_{-i}) \geq 0.
\end{equation*}
Therefore for some $a^* \in \rm{supp}(\sigma)$, the efficiency is lower bounded by $\w(a^*) \geq \frac{\lambda}{\mu} \w(\aopt)$.
\end{proof}

\section{Conclusion}
\label{sec:conc}
The main focus of this paper is on studying the limiting behavior of self-interested agents when it is not guaranteed that they converge to Nash equilibrium. As such, we consider sink equilibrium as our main solution concept, which is naturally defined as a limiting distribution of the best response process. To quantify performance of the sink equilibrium of a game, we examine the price of sinking as an appropriate approximation metric. For Nash equilibrium, smoothness has been classically used to derive guarantees for price of anarchy, the respective approximation metric. However, we show that smoothness is not enough to guarantee anything for the price of sinking. But by using a novel method of misalignment parametrization, this work provides nontrivial guarantees on the sink equilibrium. We implement our guarantees into two natural application settings and provide bounds on the resulting sink equilibrium. Future work consists of considering additional settings where sink equilibrium may appear and providing tighter approximation guarantees.

\bibliographystyle{ieeetr}
\bibliography{references.bib}

\appendix
\label{sec:appendix}

\begin{proof}[Proof of Proposition \ref{prop:ex1}]
We first characterize the expected misalignment for the setting in Example \ref{ex:env} in the following lemma.
\begin{lemma}
\label{lem:misex1}
Consider the problem outlined in Example \ref{ex:env} with $\beta_\Phi$ defined in Eq. \eqref{eq:misex1}. The expected arithmetic misalignment of the welfare function in Eq. \eqref{eq:setcov} is $\beta_\Phi$.
\end{lemma}
\begin{proof}
We derive a bound for $\frac{|\ui(a) - \w(a)|}{\w(a)}$ based on the parameters given in Example \ref{ex:env}. From the equations defining the utility and welfare,
\begin{align*}
    \frac{|\ui(a) - \w(a)|}{\w(a)} &= \frac{\left| \sum_{r \in \cup_i a_i} v_r^i - \sum_{r \in \cup_i a_i} v_r \right|}{\sum_{r \in \cup_i a_i} v_r} \\
    &\leq \frac{\sum_{r} |v_r^i - v_r|}{\sum_{r} v_r},
\end{align*}
where we use triangle inequality and $\cup_i a_i \subseteq \cl{R}$ to get the inequality on the right hand side. Observe that $\frac{\sum_i x_i}{\sum_i y_i} \leq \sum_i \frac{x_i}{y_i}$ for nonnegative $x_i$ and $y_i$. This fact coupled with linearity of expectation and $v_r \geq 0$ gives the inequality
\begin{equation*}
    \bb{E}\left[\frac{|\ui(a) - \w(a)|}{\w(a)} \right] \leq \sum_r \bb{E} \left[ | \frac{v_r^i - v_r}{v_r} | \right].
\end{equation*}
We assume that $v_r^i \sim \cl{N}[v_r + c, (d \cdot v_r)^2]$ is drawn from a normal distribution. Thus the expectation $\bb{E} \left[ | \frac{v_r^i - v_r}{v_r} | \right] = \bb{E}[\cl{N}_f(c, d^2)]$, where $\cl{N}_f$ is a folded normal distribution with mean $c$ and variance $d^2$. This holds true for any $r \in \cl{R}$, so using the equality
\begin{equation*}
\bb{E} \left[ | \frac{v_r^i - v_r}{v_r} | \right] = \left( d \sqrt{\frac{2}{\pi}}e^{-\frac{c^2}{2d^2}} + c(1 - 2 \Phi(-c/d)) \right)   
\end{equation*}
results in the bound given in Eq. \eqref{eq:misex1}.
\end{proof}

We remark that the covering problem in Example \ref{ex:env} is a submodular game \cite{paccagnan2021utility}. Under the common interest utility, the constants $\lambda_c = 1$ and $\mu_c = 2$ satisfy the inequality in Eq. \eqref{eq:comsmooth} for submodular games \cite[Example 2.6]{roughgarden2015intrinsic}. We can directly apply the guarantee given in Eq. \eqref{eq:poseart} for the arithmetic misalignment for
\begin{align*}
    \bb{E}[\pose(\G)] &\geq \bb{E}[\max(\frac{\lambda_c - 4 n \beta }{\mu_c}, 0)] \\
    &\geq \max(\frac{\lambda_c - 4 n \bb{E}[\beta] }{\mu_c}, 0),
\end{align*}
applying Jensen's inequality. Note that $\bb{E}[\beta] = \beta_\Phi$ is given in Eq. \eqref{eq:misex1} and substituting $\lambda_c = 1$ and $\mu_c = 2$ gives the lower bound.
\end{proof}

\begin{proof}[Proof of Proposition \ref{prop:ex2}]
We first characterize the expected misalignment for the setting in Example \ref{ex:radio} in the following lemma.
\begin{lemma}
\label{lem:misex2}
Consider the problem outline in Example \ref{ex:radio} with $w_{ij}^{i}/w_{ij} \in [\alpha, \alpha^{-1}]$ for all interference weight estimates. The welfare function in Eq. \eqref{eq:welfrad} is $(1 - \alpha)$-geometrically misaligned.
\end{lemma}
\begin{proof}
We derive a bound for $\frac{\ui(a)}{\w(a)}$ based on the parameters given in Example \ref{ex:env}. First, we verify the identity
\begin{equation*}
    \min_i \frac{x_i}{y_i} \leq \frac{\sum_i x_i}{\sum_i y_i} \leq \max_i \frac{x_i}{y_i}.
\end{equation*}
Let $m = \min_i \frac{x_i}{y_i}$ and $M = \max_i \frac{x_i}{y_i}$. Then $m \sum_i y_i \leq \sum_i x_i \leq M \sum_i y_i$ and the identity hold true by dividing all sides by $\sum_i y_i$. Now we can use this identity to show that
\begin{align*}
    \alpha &\leq \min_{i,j} \frac{w_{ij}^{i}}{w_{ij}} \leq \frac{\ui(a)}{\w(a)} = \frac{\sum_i \sum_{j: k_j \neq k_i} w^i_{ij}}{\sum_i \sum_{j: k_j \neq k_i} w_{ij}} \\
    &\leq \max_{i,j} \frac{w_{ij}^{i}}{w_{ij}} \leq \alpha^{-1},
\end{align*}
from the assumption that $w_{ij}^{i} \in [\alpha \cdot w_{ij}, \alpha^{-1} \cdot w_{ij}]$. Thus, we see that game is geometrically misaligned with $\beta = 1 -\alpha$.
\end{proof}

We claim that $\lambda_c = 1$ and $\mu_c = 3$ are valid constants that satisfy Eq. \eqref{eq:comsmooth} when $k=1$. Under the claim, subbing $\lambda_c$ and $\mu_c$ in Eq. \eqref{eq:posegeo} for geometric misalignment with $\beta = (1-\alpha)$ gives the final expression in Eq. \eqref{eq:misex2}.

Now we show the claim. The action set for agent $i$ can be defined as $\ac_i = \{1, 2\}$, depending on which channel agent $i$ chooses. Let $\hat{a}$ be an arbitrary joint action and $\aopt$ be the joint action that maximizes the welfare. We partition the agent set $\p$ with the following subsets
\begin{align*}
    D &= \{i \in \p: \hat{a}_i = 1 \text{ and } \aopt_i = 2 \}, \\
    O &= \{i \in \p: \hat{a}_i = 2 \text{ and } \aopt_i = 1 \}, \\
    B &= \{i \in \p: \hat{a}_i = 1 \text{ and } \aopt_i = 1 \}, \\
    N &= \{i \in \p: \hat{a}_i = 2 \text{ and } \aopt_i = 2 \}.
\end{align*}
Additionally, for ease of notation, given subsets $S_1, S_2 \subset \p$, we can define $w(S_1, S_2) =  \sum_{i \in S_1} \sum_{j \in S_2} w_{ij} + w_{ji}$. Under these definitions, we have that the welfare function can be written as
\begin{align*}
    \w(\hat{a}) &= \sum_i \sum_{j : \hat{a}_i \neq \hat{a}_j} w_{ij} \\
    &= w(B, N) + w(D, N) + w(O, B) + w(O, D).
\end{align*}
Similarly, it can be verified that $\w(\aopt) = w(B, N) + w(O, N) + w(D, B) + w(D, O)$. Next, we rewrite the sum of deviations as 
\begin{align*}
    &\sum_{i \in \p} \w(a)  - \w(\aopt_i, \hat{a}_{-i}) = \\
    &w(D, N) + w(O, B) + 2 w(O, D) - \\
    &w(D, D) - w(O, O) - w(D, B) - w(O, N).
\end{align*}
From these definitions, we have the following set of inequalities,
\begin{align*}
    \w(\aopt) + \sum_{i \in \p} \w(a)  - \w(\aopt_i, a_{-i}) &\leq \\
    3 w(B, O) + w(B, N) + w(D, N) + & \\
    w(O, B) - w(D, D) - w(O, O) &\leq \\
    3 w(B, O) + 3 w(B, N) + 3 w(D, N) + 3 w(O, B) &\leq 3 \w(\hat{a}).
\end{align*}
Since $\hat{a}$ was chosen arbitrarily, we see that Eq. \eqref{eq:comsmooth} is satisfied for $\lambda_c = 1$ and $\mu_c = 3$ for all joint actions. Thus the guarantee in Eq. \eqref{eq:misex2} holds.
\end{proof}

\end{document}